\documentclass[11pt]{article}
\usepackage[a4paper, margin=27mm]{geometry}

\usepackage{amsmath,amsthm}  

\usepackage{newtxtext, newtxmath}
\usepackage{microtype}
\usepackage{bm}
\usepackage{mathtools}
\usepackage{array}

\usepackage{graphicx} 

\theoremstyle{plain}
\newtheorem{theorem}{Theorem}

\newtheorem{proposition}{Proposition}
\newtheorem{lemma}{Lemma}

\theoremstyle{definition}

\begin{document}

\title{\vspace{-0.6cm}\textbf{Identifiability and improper solutions in the probabilistic partial least squares regression with unique variance}}

\author{Takashi Arai\thanks{\texttt{takashi-arai@sci.kj.yamagata-u.ac.jp}} \vspace{0.2cm} \\ 
Faculty of Science, Yamagata University, Yamagata 990-8560, Japan}


\date{\vspace{-0.5cm}}

\maketitle

\begin{abstract}
    This paper addresses theoretical issues associated with probabilistic partial least squares (PLS) regression.
    As in the case of factor analysis, the probabilistic PLS regression with unique variance suffers from the issues of improper solutions and lack of identifiability, both of which causes difficulties in interpreting latent variables and model parameters.
    Using the fact that the probabilistic PLS regression can be viewed as a special case of factor analysis, we apply a norm constraint prescription on the factor loading matrix in the probabilistic PLS regression, which was recently proposed in the context of factor analysis to avoid improper solutions.
    Then, we prove that the probabilistic PLS regression with this norm constraint is identifiable.
    We apply the probabilistic PLS regression to data on amino acid mutations in Human Immunodeficiency Virus (HIV) protease to demonstrate the validity of the norm constraint and to confirm the identifiability numerically.
    Utilizing the proposed constraint enables the visualization of latent variables via a biplot.
    We also investigate the sampling distribution of the maximum likelihood estimates (MLE) using synthetically generated data.
    We numerically observe that MLE is consistent and asymptotically normally distributed.
\end{abstract}

\section{Introduction}
Partial least squares (PLS) regression is a multivariate statistical method that constructs lower-dimensional composite variables, also regarded as latent variables, by taking linear combinations of explanatory variables, and uses these composite variables to predict objective variables.
PLS regression can handle ill-posed regression where the number of explanatory variables $p$ exceeds the number of observations $n$, as well as situations where strong multicollinearity among explanatory variables exists.
Therefore, PLS regression can be used as an alternative to regularized regression methods.
PLS regression is extensively used in chemometrics and is considered a standard method used for ill-posed multivariate linear regression.
PLS regression is also used in the life sciences and in multivariate statistical process control (MSPC) in industrial processes, where ill-posed problems and multicollinearity often arise.
Specifically, in the life sciences, PLS regression is applied to genomic data~\cite{Boulesteix2006}, while in industrial applications, it is used for process monitoring and soft sensing~\cite{Zheng2016}.

The classical PLS regression has many variations, e.g., PLS1, PLS2, SIMPLS, NIPALS, etc., depending on the objective functions and algorithms in estimating model parameters~\cite{Boulesteix2006}.
Due to the many algorithmic variants of PLS regression, it is difficult for practitioners of data analysis to understand the principles underlying PLS.
Although PLS regression has been applied in a variety of fields, it has not become as widely used as the more popular method, such as principal component analysis, possibly due to confusion stemming from its many algorithmic variants.
In light of this issue, it is desirable to formulate PLS within a probabilistic framework in order for the model to become more widely adopted and recognized as an established method.
In a probabilistic framework, the concept of likelihood can serve as a unified objective function for parameter estimation, thereby avoiding the confusion caused by the existence of algorithmic variants in classical PLS regression.
Furthermore, the probabilistic PLS model can naturally handle missing values in the dataset, which is very common in practice.

As a probabilistic extension of PLS regression, probabilistic PLS regression has been proposed~\cite{Gustafsson2001, Murphy2012}.
This model can handle noisy explanatory variables, which are commonly observed in experimental data, but the model assumes that the variance of the observational noise is the same across all features.
The probabilistic PLS model has been further extended to include unique variance~\cite{Zheng2016}, independent variances specific to each variable, which makes the model more practical.
However, the introduction of unique variance causes the problem of improper solutions, a problem that also arises in factor analysis~\cite{Gerbing1985, Gerbing1987}.
Specifically, maximum likelihood estimates of certain unique variances may become close to zero, resulting in extremely large absolute values of the estimated latent variables and making their interpretation difficult.
Furthermore, in order to appropriately interpret the latent variables in PLS regression, the identifiability problem of model parameters should be addressed.
In fact, neither classical PLS regression nor the probabilistic PLS regression are identifiable unless appropriate constraints are imposed.
The improper solutions and nonidentifiability in PLS regression do not necessarily affect the predictive performance of the model, and therefore may not be considered problematic in fields where the primary concern of the analysis is solely on predictive performance.
In fact, in the field of chemometrics, the primary objective is to improve the predictive accuracy for objective variables.
As a result, the interpretation of latent variables has not been emphasized, and the issues related to identifiability and improper solutions have received little attention~\cite{Zheng2016}.
However, in applications in the life sciences such as genetic analysis, where the interpretation of latent variables is essential, interpreting latent variables and model parameters is essential for understanding the underlying molecular mechanisms.
In addition, the latent scores obtained by PLS regression are often used as explanatory variables when applying other predictive models.
For example, in partial least squares discriminant analysis (PLSDA), which is an extension of PLS regression to discriminant analysis, the latent scores are used as explanatory variables for the discriminant analysis~\cite{Chung2010}.
Therefore, for the effective use of the PLS model, it is important to address issues related to identifiability and improper solutions.

In this paper, we address the issues related to improper solutions and identifiability in the probabilistic PLS regression.
The identifiability of probabilistic PLS regression has been addressed in the previous study~\cite{Bouhaddani2018}; however, the model does not incorporate unique variance.
In addition, severe restrictions on the number of latent dimensions and constraints on the structure of model parameters are imposed to ensure identifiability, which impairs the representational capability and predictive performance of the model.
The model considered in this paper is identical to the probabilistic PLS regression with unique variance proposed in the previous study~\cite{Zheng2016}.
To avoid improper solutions, we adopt a recently proposed norm constraint that enforces equal row norms of the factor loading matrix across all observed variables~\cite{Arai2022}.
We also establish the identifiability of the model under this norm constraint.

This paper is organized as follows.
In Sec.~\ref{sec:summary}, we summarize the probabilistic PLS model and propose methods for addressing the issues of improper solutions and identifiability.
In Sec.~\ref{sec:numerical_validation}, we validate the proposed method through numerical experiments.
We apply probabilistic PLS regression to real data on amino acid mutations in HIV protease and show a biplot.
We also investigate the sampling distribution of the maximum likelihood estimates of the model parameters.
Sec.~\ref{sec:conclusion} is devoted to conclusion.
A proof of model identifiability is presented in the appendices.

\section{Statement of the result} \label{sec:summary}
Let $\mathbf{x}$ and $\mathbf{y}$ be continuous explanatory and objective variables and are represented by $p_x$- and $p_y$-dimensional column vectors, respectively.
We introduce continuous latent variables $\mathbf{u}$ and $\mathbf{v}$, which are column vectors of dimensions $p_u$ and $p_v$.
The latent variables $\mathbf{u}$ contribute to both the explanatory and objective variables; the dependence between explanatory and objective variables is introduced by these latent variables.
In contrast, the latent variables $\mathbf{v}$ contribute only to the explanatory variables and are introduced to capture the input-specific variations in the explanatory variables.
We refer to the subspace spanned by $\mathbf{u}$ as the shared latent space, and the subspace spanned by $\mathbf{v}$ as the unshared latent space or input-specific latent space.
Then, the prior distribution of the latent variables and the conditional distribution given the latent variables are represented by normal distributions, based on the conditional independence assumption:
\begin{align}
	p(\mathbf{z}) =& p(\mathbf{u}) p(\mathbf{v}) = \mathcal{N}(\mathbf{u} \mid \bm{\mu}_u, \Sigma_u) \mathcal{N}(\mathbf{v} \mid \bm{\mu}_v, \Sigma_v), \\
	p(\mathbf{x}, \mathbf{y} \mid \mathbf{z}) =& p(\mathbf{x} \mid \mathbf{z}) p(\mathbf{y} \mid \mathbf{z}), \\
	p(\mathbf{x} \mid \mathbf{z}) =& \mathcal{N}(\mathbf{x} \mid \bm{\mu}_x + W_{xu} (\mathbf{u} - \bm{\mu}_u) + W_{xv} (\mathbf{v} - \bm{\mu}_v), \Psi_x), \label{eq:x_given_z} \\
	p(\mathbf{y} \mid \mathbf{z}) =& \mathcal{N}(\mathbf{y} \mid \bm{\mu}_y + W_{yu} (\mathbf{u} - \bm{\mu}_u), \Psi_y), \label{eq:y_given_z}
\end{align}
where the parameters denoted by $\bm{\mu}$ are corresponding mean vectors, and the diagonal matrices $\Sigma_u$, $\Sigma_v$, $\Psi_x$, $\Psi_y$ are corresponding covariance matrices of the normal distributions.
$W_{xu}$, $W_{xv}$, and $W_{yu}$ are $p_x \times p_u$, $p_x \times p_v$, and $p_y \times p_u$ matrices, respectively.
$\Psi_x$ and $\Psi_y$ are referred to as the unique variance and represent the variance of observed variables that is not accounted for by the latent variables.
When $\Psi_y = \Psi_x = \sigma^2 I$, where $I$ is the identity matrix, the model reduces to the formulation presented in the textbook~\cite{Murphy2012}.

From the above expressions, we can read the expression for the complete data likelihood including the latent variables as 
\begin{align}
	& p(\mathbf{x}, \mathbf{y}, \mathbf{z}) \notag \\
    =&
	\mathcal{N}(\mathbf{x} \mid \bm{\mu}_x + W_{xu}\mathbf{u} + W_{xv} \mathbf{v}, \Psi_x)
	\mathcal{N}(\mathbf{y} \mid \bm{\mu}_y + W_{yu}\mathbf{u}, \Psi_y)
	\mathcal{N}(\mathbf{u} \mid \bm{\mu}_u, \Sigma_u) \mathcal{N}(\mathbf{v} \mid \bm{\mu}_v, \Sigma_v), \notag \\
	=&
	\frac{1}{(2\pi)^{p/2}} \exp\biggl\{
	-\frac{1}{2} (\mathbf{x} - \bm{\mu}_x)^T \Psi_x^{-1}(\mathbf{x} - \bm{\mu}_x)
	-\frac{1}{2} (\mathbf{y} - \bm{\mu}_y)^T \Psi_y^{-1}(\mathbf{y} - \bm{\mu}_y) \notag \\
	& + \frac{1}{2} \mathbf{x}^T \Psi_x^{-1} (W_{xu} \mathbf{u} + W_{xv} \mathbf{v}) \times 2
	+ \frac{1}{2} \mathbf{y}^T \Psi_y^{-1} W_{yu} \mathbf{u} \times 2
	- \frac{1}{2} \mathbf{u}^T (\Sigma_u^{-1} + W_{ux} \Psi_x^{-1} W_{xu} + W_{uy} \Psi_y^{-1} W_{yu}) \mathbf{u} \notag \\
	& - \frac{1}{2} \mathbf{v}^T (\Sigma_v^{-1} + W_{vx} \Psi_x^{-1} W_{xv} ) \mathbf{v}
	- \frac{1}{2} \mathbf{u}^T W_{ux} \Psi_x^{-1} W_{xv} \mathbf{v} \times 2
	\biggr\},
\end{align}
where we have defined $W_{xu}^T = W_{ux}$, $W_{xv}^T = W_{vx}$, and $W_{yu}^T = W_{uy}$, and $T$ stands for matrix transposition.
The above expression implies that the distribution $p(\mathbf{x}, \mathbf{y}, \mathbf{z})$ is the normal distribution with the covariance matrix given by
\begin{align}
    p(\mathbf{x}, \mathbf{y}, \mathbf{z}) =& \mathcal{N}(\mathbf{x}, \mathbf{y}, \mathbf{z} \mid \bm{\mu}, \Sigma), \\
    \bm{\mu} =& [\bm{\mu}_x; \bm{\mu}_y; \bm{\mu}_z], \\
    \Lambda =&
    \begin{bmatrix} \Lambda_x & \Lambda_{xy} & \Lambda_{xz} \\
                    \Lambda_{yx} & \Lambda_y & \Lambda_{yz} \\
                    \Lambda_{zx} & \Lambda_{zy} & \Lambda_z
    \end{bmatrix} = 
    \begin{bmatrix} \Psi^{-1} & - \Psi^{-1} W \\ - W^T \Psi^{-1} & \Sigma_z^{-1} + W^T \Psi^{-1} W \end{bmatrix}, \notag \\
    =&
    \begin{bmatrix}
		\Psi_x^{-1} & O & - \Psi_x^{-1} W_{xu} & - \Psi_x^{-1} W_{xv} \\
		O & \Psi_y^{-1} & - \Psi_y^{-1} W_{yu} & O \\
		- W_{ux} \Psi_x^{-1} & - W_{uy}\Psi_y^{-1} & \Sigma_u^{-1} + W_{ux} \Psi_x^{-1} W_{xu} + W_{uy} \Psi_y^{-1} W_{yu} & W_{ux} \Psi_x^{-1} W_{xv} \\
		- W_{vx}\Psi_x^{-1} & O & W_{vx} \Psi_x^{-1} W_{xu} & \Sigma_v^{-1} + W_{vx} \Psi_x^{-1} W_{xv}
	\end{bmatrix},
\end{align}
\begin{align}
	\Sigma = & \Lambda^{-1}, \notag \\
	=&
    \begin{bmatrix} \Sigma_x & \Sigma_{xy} & \Sigma_{xz} \\
                    \Sigma_{yx} & \Sigma_y & \Sigma_{yz} \\
                    \Sigma_{zx} & \Sigma_{zy} & \Sigma_z
    \end{bmatrix} = 
	\begin{bmatrix} \Psi + W \Sigma_z W^T & W \Sigma_z \\ \Sigma_z W^T & \Sigma_z \end{bmatrix}, \notag \\
	=&
	\begin{bmatrix}
		\Psi_x + W_{xu} \Sigma_{u} W_{ux} + W_{xv} \Sigma_{v} W_{vx} & W_{xu} \Sigma_{u} W_{uy} & W_{xu} \Sigma_{u} & W_{xv} \Sigma_{v} \\
		W_{yu} \Sigma_{u} W_{ux} & \Psi_y + W_{yu} \Sigma_{u} W_{uy} & W_{yu} \Sigma_{u} & O \\
		\Sigma_{u} W_{ux} & \Sigma_{u} W_{uy} & \Sigma_{u} & O \\
		\Sigma_{v} W_{vx} & O & O & \Sigma_{v}
	\end{bmatrix}, \\
    \Psi =&
    \begin{bmatrix} \Psi_x & O \\ O & \Psi_y \end{bmatrix}, \hspace{0.3cm}
    \Sigma_z =
    \begin{bmatrix} \Sigma_u & O \\ O & \Sigma_v \end{bmatrix}, \hspace{0.3cm}
    W = \begin{bmatrix}
        W_{xu} & W_{xv} \\ 
        W_{yu} & O
    \end{bmatrix}
    = \begin{bmatrix} W_x \\ W_y \end{bmatrix},
\end{align}
where semicolon denotes a vertical concatenation of vectors $[\bm{\mu}_x; \bm{\mu}_y] = [\bm{\mu}_x^T, \bm{\mu}_y^T]^T$, and $O$ is a matrix whose elements are all zero.
The above PLS regression can be viewed as a special case of factor analysis, in which the factor loading matrix is block-partitioned.
Hence, we refer to the parameter $W$ as the factor loading matrix, or simply the loading matrix.

Then, the observed distribution $p(\mathbf{x}, \mathbf{y})$ is also expressed analytically:
\begin{align}
	p(\mathbf{x}, \mathbf{y}) =& \mathcal{N}([\mathbf{x}; \mathbf{y}] \mid [\bm{\mu}_x; \bm{\mu}_y], \Psi + W \Sigma_z W^T), \label{eq:likelihood} \\
	p(\mathbf{x}) =& \mathcal{N}( \mathbf{x} \mid \bm{\mu}_x, \Sigma_x), \\
	p(\mathbf{y} \mid \mathbf{x}) =& \mathcal{N}(\mathbf{y} \mid \bm{\mu}_y + W_{yu} \Sigma_u W_{ux} \Sigma_x^{-1} (\mathbf{x} - \bm{\mu}_x), \Sigma_{y|x}), \label{eq:y_given_x}\\
    \Sigma_{y|x} =& \Psi_y + W_{yu} \Sigma_u W_{uy} - W_{yu} \Sigma_u W_{ux} \Sigma_x^{-1} W_{xu} \Sigma_u W_{uy}.
\end{align}
Although it is called probabilistic PLS regression, it is in fact a probabilistic generative model.
The Eq.~(\ref{eq:likelihood}) implies that the probabilistic PLS regression merely represents the covariance matrix as a low-rank perturbation of the diagonal matrix $\Psi$.
While the Eq.~(\ref{eq:y_given_x}) can be viewed as applying linear regression using composite variables $\mathbf{m}^{(x)} \equiv \Sigma_u W_{ux} \Sigma_x^{-1} (\mathbf{x} - \bm{\mu}_x)$, it can also be interpreted as expressing the regression coefficient matrix $B_{yx}$ in a low-rank form $B_{yx} = \Sigma_u W_{ux} \Sigma_x^{-1}$.
Thus, the probabilistic PLS regression can be implemented without explicitly considering underlying latent variables.

Furthermore, the conditional distributions of the latent variables given the observed variables are also obtained as follows:
\begin{align}
    p(\mathbf{z} \mid \mathbf{x}, \mathbf{y}) =&
    \mathcal{N}(\mathbf{z} \mid \mathbf{m}^{(x,y)}, \Sigma_{z|(xy)}), \\
    \Sigma_{z|(xy)} =&
    \Lambda_z^{-1} = [\Sigma_z^{-1} + W^T \Psi^{-1} W]^{-1}, \\
    \mathbf{m}^{(x,y)} =&
    \bm{\mu}_{z|(xy)}, \notag \\
    \equiv&
    \bm{\mu}_z + \Sigma_{z(xy)} \Sigma_{(xy)(xy)}^{-1} \begin{bmatrix} \mathbf{x} - \bm{\mu}_x \\ \mathbf{y} - \bm{\mu}_y \end{bmatrix}, \notag \\
    =&
    \bm{\mu}_z + \Sigma_{z|(xy)} W^T \Psi^{-1} \begin{bmatrix} \mathbf{x} - \bm{\mu}_x \\ \mathbf{y} - \bm{\mu}_y \end{bmatrix},
\end{align}
\begin{align}
    p(\mathbf{z} \mid \mathbf{x}) =&
    \mathcal{N}(\mathbf{z} \mid \mathbf{m}^{(x)}, \Sigma_{z|x}), \\
    \Sigma_{z|x} =&
    [\Sigma_{(yz)|x}]_{zz} = [\Lambda_{(yz)(yz)}^{-1}]_{zz}, \notag \\
    =& [\Lambda_{z} - \Lambda_{zy}\Lambda_{yy}^{-1} \Lambda_{yz}]^{-1} = [\Sigma_z^{-1} + W_x^T \Psi_x^{-1} W_x]^{-1}, \\
    \mathbf{m}^{(x)} =&
    \bm{\mu}_{z|x}, \notag \\
    \equiv&
    \bm{\mu}_z + \Sigma_{zx} \Sigma_{x}^{-1} (\mathbf{x} - \bm{\mu}_x), \notag \\
    =&
    \bm{\mu}_z - [\Lambda_{(yz)(yz)}^{-1}]_{z(yz)} \Lambda_{(yz)x} (\mathbf{x} - \bm{\mu}_x), \notag \\
    =&
    \bm{\mu}_z - \begin{bmatrix} - [\Lambda_{(yz)(yz)}^{-1}]_{zz} \Lambda_{zy} \Lambda_{yy} \;\;\;\; & [\Lambda_{(yz)(yz)}^{-1}]_{zz} \end{bmatrix} \begin{bmatrix} O \\ -W_x^T \Psi_x^{-1} \end{bmatrix} (\mathbf{x} - \bm{\mu}_x), \notag \\
    =&
    \bm{\mu}_z + \Sigma_{z|x} W_x^T \Psi_x^{-1} (\mathbf{x} - \bm{\mu}_x),
\end{align}
\begin{align}
    p(\mathbf{z} \mid \mathbf{y}) =& \mathcal{N}(\mathbf{z} \mid \mathbf{m}^{(y)}, \Sigma_{z|y}), \\
    \Sigma_{z|y} =& [\Lambda_{(xz)(xz)}^{-1}]_{zz} = [\Lambda_{z} - \Lambda_{zx} \Lambda_{xx}^{-1} \Lambda_{xz}]^{-1}, \notag \\
    =& 
    \begin{bmatrix} \Sigma_u^{-1} + W_{uy} \Psi_y^{-1} W_{yu} & O \\ O & \Sigma_v^{-1} \end{bmatrix}^{-1}
    =
    \begin{bmatrix} [\Sigma_u^{-1} + W_{uy} \Psi_y^{-1} W_{yu}]^{-1} & O \\ O & \Sigma_v \end{bmatrix}, \\
    \mathbf{m}^{(y)} =&
    \bm{\mu}_{z|y}, \notag \\
    \equiv&
    \bm{\mu}_z + \Sigma_{zy} \Sigma_{y}^{-1} (\mathbf{y} - \bm{\mu}_y), \notag \\
    =&
    \bm{\mu}_z - [\Lambda_{(xz)(xz)}^{-1}]_{z(xz)} \Lambda_{(xz) y} (\mathbf{y} - \bm{\mu}_y), \notag \\
    =&
    \bm{\mu}_z - \begin{bmatrix} - [\Lambda_{(xz)(xz)}^{-1}]_{zz} \Lambda_{zx} \Lambda_{xx} \;\;\;\; & [\Lambda_{(xz)(xz)}^{-1}]_{zu} \;\;\;\;  & [\Lambda_{(xz)(xz)}^{-1}]_{zv} \end{bmatrix}
    \begin{bmatrix} O \\ - W_{uy} \Psi_y^{-1} \\ O \end{bmatrix} (\mathbf{y} - \bm{\mu}_y), \notag \\
    =&
    \begin{bmatrix} \bm{\mu}_u + [\Sigma_u^{-1} + W_{uy} \Psi_y^{-1} W_{yu}]^{-1} W_{uy} \Psi_y^{-1} (\mathbf{y} - \bm{\mu}_y) \\ \bm{\mu}_v \end{bmatrix}
    =
    \begin{bmatrix} \bm{\mu}_u + \Sigma_{u|y} W_{uy} \Psi_y^{-1} (\mathbf{y} - \bm{\mu}_y) \\ \bm{\mu}_v \end{bmatrix},
\end{align}
where the submatrix such as $\Lambda_{x(yz)}$ is defined by
\begin{align}
    \Lambda_{x(yz)} \equiv \begin{bmatrix} \Lambda_{xy} & \Lambda_{xz} \end{bmatrix}, \hspace{0.3cm} \text{etc.},
\end{align}
$[\Lambda_{(xz)(xz)}]_{zz}$ denotes the submatrix of $\Lambda_{(xz)(xz)}$ with both rows and columns indexed by $z$, the inverse notation $\Lambda_{xx}^{-1} \equiv [\Lambda_{xx}]^{-1}$ means the inverse matrix of the submatrix $\Lambda_{xx}$, and $\Sigma_{x|y} \equiv \Sigma_{xx} - \Sigma_{xy} \Sigma_{yy}^{-1} \Sigma_{yx}$ is the Schur complement of the matrix $\Sigma_{(xy)(xy)}$ with respect to $\Sigma_{yy}$.
Hereafter, we refer to $\mathbf{m}^{(x,y)}$, $\mathbf{m}^{(x)}$ and $\mathbf{m}^{(y)}$ as the factor scores.

Although the predictive distribution given explanatory variables is represented by Eq.~(\ref{eq:y_given_x}), it is informative to express it in terms of factor scores for interpreting the meaning of the latent space.
The predictive distribution can be approximated by
\begin{align}
    p(\mathbf{y} \mid \mathbf{x}) =& 
    \int d \mathbf{z} p(\mathbf{y} \mid \mathbf{z}) p(\mathbf{z} \mid \mathbf{x}), \notag \\
    \simeq &
    p(\mathbf{y} \mid \hat{\mathbf{z}}), \hspace{1cm} \hat{\mathbf{z}} = \mathrm{argmax} \, p(\mathbf{z} \mid \mathbf{x}) = \mathbf{m}^{(x)}.
\end{align}
Therefore, the factor scores $\mathbf{m}^{(x)}$ can be regarded as composite explanatory variables constructed by compressing the explanatory variables through a linear combination, and the factor loading matrix $W_{yu}$ can be interpreted as the regression coefficient for the variables $\mathbf{m}^{(x)}$.

In the above model, the observed distribution $p(\mathbf{x}, \mathbf{y})$ does not depend on the parameters $\bm{\mu}_u$ and $\bm{\mu}_v$.
In addition, the covariance matrix of the latent variables $\Sigma_z$ can be renormalized into the factor loading matrix as $W' = W \Sigma_z^{1/2}$.
Therefore, hereafter we set $\bm{\mu}_u = \bm{0}$, $\bm{\mu}_v = \bm{0}$, $\Sigma_u = I$, and $\Sigma_v = I$ without loss of generality.
Furthermore, in this paper, we impose the condition that each column of the factor loading submatrices $W_{yu}$ and $W_{xv}$ be orthogonal, respectively.

\subsection{Improper solutions}
It is natural to introduce unique variance, which allows each observed variable given the latent variables to have specific variance, as shown in Eqs.~(\ref{eq:x_given_z}, \ref{eq:y_given_z}).
The same attempt has already been made in the previous study~\cite{Zheng2016}.
However, the introduction of unique variance leads to improper solutions in the maximum likelihood estimation of model parameters as in the case of factor analysis.
In fact, the maximum likelihood estimation of Eq.~(\ref{eq:likelihood}) may cause certain elements of the unique variances to approach zero.
Then, the corresponding elements of the factor scores, such as $\mathbf{m}^{(x)}$, take large values, and the interpretation of the latent variables becomes difficult.
Furthermore, improper solutions cause difficulties in numerical parameter estimation.
Specifically, when the parameters become trapped in a region corresponding to an improper solution, convergence of numerical optimization becomes slow.

To avoid improper solutions in factor analysis, we recently proposed a prescription~\cite{Arai2022}, which imposes a constraint that the row norms of the scaled factor loading matrix are identical for all features.
In this paper, we propose a similar norm constraint for probabilistic PLS regression as follows:
\begin{align}
    & \bar{W} \equiv \Psi^{-1/2} W, \hspace{0.5cm} \mathrm{diag}(\bar{W} \bar{W}^T) = c^2 I, \hspace{0.5cm} \tilde{W} = \frac{1}{c} \bar{W}, \\
    & \hat{W} \equiv \mathrm{diag}(\Sigma_x)^{-1/2} W, \hspace{0.5cm} \mathrm{diag}(\hat{W} \hat{W}^T) = h^2 I, \\
    & h^2 = \frac{c^2}{1 + c^2}, \hspace{0.5cm} c^2 = \frac{h^2}{1 - h^2}, \hspace{0.5cm} c \ge 0, \hspace{0.5cm} 0 \le h \le 1,
\end{align}
where $\mathrm{diag}()$ denotes the diagonal elements or diagonal matrix depending on the context, and the row norms of the normalized factor loading matrix $\tilde{W}$ are normalized to one.
Throughout this paper, we impose this norm constraint in our PLS regression.

This norm constraint is also useful in taking a classical limit of probabilistic PLS regression.
The classical limit is obtained by taking the variances of the observational noise to zero $\Psi \rightarrow O$.
This limit is equivalent to taking $c \rightarrow \infty$ and $h \rightarrow 1$.
Since the unique variances are expressed as $\Psi_x = \mathrm{diag}(\Sigma_x) / (1 + c^2)$, and in this limit, the factor loading matrix can be approximated by $W_x \simeq \mathrm{diag}(\Sigma_x)^{1/2} \tilde{W}_x$, the factor scores has the following classical limit:
\begin{align}
    \mathbf{m}^{(x)} \simeq&
    \frac{I}{I + (1+c^2) \tilde{W}_x^T \tilde{W}_x} \tilde{W}_x^T \mathrm{diag}(\Sigma_x)^{1/2} (1+c^2) \mathrm{diag}(\Sigma_x)^{-1} (\mathbf{x} - \bm{\mu}_x), \notag \\
    \simeq&
    (\tilde{W}_x^T \tilde{W}_x)^{-1} \tilde{W}_x^T \mathrm{diag}(\Sigma_x)^{-1/2} (\mathbf{x} - \bm{\mu}_x).
\end{align}
Therefore, the classical limit of probabilistic PLS regression can be viewed as the projection of the standardized explanatory variables $\mathrm{diag}(\Sigma_x)^{-1/2} (\mathbf{x} - \bm{\mu}_x)$ onto the subspace spanned by the column vectors of the normalized factor loading matrix $\tilde{W}_x$.

\subsection{Identifiability}
As in the case of factor analysis, probabilistic PLS regression has rotational and sign-flip degrees of freedom in the latent space.
In fact, the likelihood function expressed as Eq.~(\ref{eq:likelihood}) is invariant under the following independent rotations and sign-flip transformations of the shared and unshared latent subspaces:
\begin{align}
	R = \begin{bmatrix}
			R_{u} & O \\
			O & R_{v}
		\end{bmatrix},
    \hspace{0.5cm}
    \Delta = \mathrm{diag}(\bm{\delta}), \hspace{0.5cm}
    \delta_{i} \in \{ -1, 1 \}, \hspace{0.3cm} \{i=1,2,\dots, p_u + p_v\}, \label{eq:rotation}
\end{align}
where $R_u$ and $R_v$ are $p_u$- and $p_v$-dimensional rotation matrix, and satisfy $R_u^T=R_u^{-1}$ and $R_v^T = R_v^{-1}$, respectively.
In order for the model to be identifiable, the rotational and reflectional degrees of freedom must be fixed.
The following theorem states the conditions for the model to be identifiable.

\begin{theorem} \label{theorem:pls_identifiability}
    Let the observed variables $\mathbf{x}$ and $\mathbf{y}$ be generated by the probabilistic PLS regression model of Eq.~(\ref{eq:likelihood}).
    If the numbers of the latent dimensions satisfy $p_u + p_v < p_x + p_y$, $p_v \le p_x$, and $p_u \le p_y$, and for the scaled factor loading matrix defined by 
    \begin{align}
        \hat{W} = \mathrm{diag}(\Sigma)^{-1/2} W = 
        \begin{bmatrix}
            \hat{W}_{xu} & \hat{W}_{xv} \\
            \hat{W}_{yu} & O
        \end{bmatrix},
    \end{align}
    their row norms have the same value for all features, $\mathrm{diag}(\hat{W} \hat{W}^T) = H^2 = h^2 I$, and their value satisfies $(0 \le h \le 1)$, and for the scaled factor loading submatrices $\hat{W}_{yu}$ and $\hat{W}_{xv}$ the symmetric matrices $\hat{W}_{yu} \hat{W}_{uy}$ and $\hat{W}_{xv} \hat{W}_{vx}$ are diagonalized and their nonzero nondegenerate eigenvalues are sorted in descending order, respectively, i.e.,
    \begin{align}
        \hat{W}_{yu} \hat{W}_{uy} =&
        \mathrm{diag}(\hat{\bm{\omega}}_{(y,u)}^2) = \hat{\Omega}_{(y,u)}^2, \hspace{0.3cm} [\hat{\omega}_{(y,u)}]_i > [\hat{\omega}_{(y,u)}]_j, \hspace{0.3cm} (i < j),
        \hspace{0.3cm} i, j \in \{1,2,\dots, p_u\}, \\
        \hat{W}_{xv} \hat{W}_{vx} =&
        \mathrm{diag}(\hat{\bm{\omega}}_{(x,v)}^2) = \hat{\Omega}_{(x,v)}^2, \hspace{0.3cm} [\hat{\omega}_{(x,v)}]_i > [\hat{\omega}_{(x,v)}]_j, \hspace{0.3cm} (i < j),
        \hspace{0.3cm} i, j \in \{1,2,\dots, p_v\},
    \end{align}
    and the row sums of the submatrices $\hat{W}_{yu}$ and $\hat{W}_{xv}$ are nonnegative,
    \begin{align}
        \sum_{i=1}^{p_y} [\hat{W}_{yu}]_{ij} \ge & 0, \hspace{0.5cm} j \in \{1, 2, \dots, p_u \}, \\
        \sum_{i=1}^{p_x} [\hat{W}_{xv}]_{ij} \ge & 0, \hspace{0.5cm} j \in \{1, 2, \dots, p_v \},
    \end{align}
    where $[\hat{W}_{yu}]_{ij}$ denotes the $(i,j)$-th element of $\hat{W}_{yu}$.
    Then $\Sigma_1 = \Sigma_2$ means $h_1 = h_2$ and $W_1 = W_2$.
\end{theorem}
The proof of the theorem is given in Appendix~\ref{sec:proof_pls}.
An important implication of the theorem is that when the objective variable is univariate, model identifiability requires that the shared latent space be one-dimensional.

Here, we note the difference between the proposed model and the previous identifiable model studied in Ref.~\cite{Bouhaddani2018}.
The model of the previous study corresponds to the following parameters when expressed using the notation in this paper:
\begin{align}
    \Sigma_u \ne&
    \sigma_u^2 I, \hspace{0.5cm} \Sigma_v = \sigma_v^2 I, \hspace{0.5cm} \Psi_x = \sigma_x^2 I, \hspace{0.5cm} \Psi_y = \sigma_y^2 I, \\
    W =&
    \begin{bmatrix} W_{xu} & W_{xv} \\ W_{yu} & O \end{bmatrix}
    = \begin{bmatrix} V B^{-1} & \;  - V B^{-1} \\ C & O \end{bmatrix}, \label{eq:loading_previous}
\end{align}
where the numbers of the shared and unshared latent dimensions must be the same $p_u = p_v$, $V$ is a $p_x \times p_u$ matrix, and $B$ is a $p_u$-dimensional diagonal matrix.
In the previous study, the above model is shown to be identifiable up to the sign of the latent space under the following three conditions:
\begin{align}
    & 0 < p_u = p_v < \mathrm{min}(p_x, p_y), \label{eq:dimension_constraint_previous} \\
    & \mathrm{diag}(B) > \bm{0}, \\
    & V^T V = C^T C = I. \label{eq:normalization_previous}
\end{align}
However, such an identifiable model impairs representational capability.
In fact, first, this model cannot accommodate the difference in the scale of the observed variables $(\mathbf{x}, \mathbf{y})$, because it does not include unique variances.
Next, the model imposes the constraint that the loading matrices for the shared and unshared latent subspaces of the explanatory variables, $W_{xu}$ and $W_{xv}$, must behave identically, which limits the representational capability compared to the model considered in this paper.
Furthermore, the condition on the dimensionality of the latent space, as shown in Eq.~(\ref{eq:dimension_constraint_previous}), limits the representability of the model when there is a large asymmetry between the numbers of explanatory and objective variables.
This is because, in such cases, the condition prevents the use of a sufficiently large number of latent dimensions.
Lastly, the orthonormality condition described in Eq.~(\ref{eq:normalization_previous}) is a constraint that cannot be achieved with rotational degrees of freedom in the shared and unshared latent space, again compromising the representational capability of the model.
The reduction in representational capability due to this parameter and dimensionality constraints leads to a decrease in predictive performance.
The decrease in predictive performance is numerically demonstrated in Sec.~\ref{sec:numerical_validation}.

\subsection{Contribution ratio} \label{sec:contribution_ratio}
As in probabilistic principal component analysis and factor analysis, the contribution ratio of the latent space in probabilistic PLS regression can be defined based on the concept of explained variance.
In our previous study, we defined the explained variance in factor analysis by the variance of the denoised observed variables when the latent variables follow a standard normal distribution $\mathbf{z} \sim \mathcal{N}(\bm{0}, I)$~\cite{Arai2021}.
We define the explained variance in probabilistic PLS regression in a similar manner.
We define the contribution ratio by the reproducibility of the variation of the denoised objective variables when the latent variables follow a standard normal distribution.

From the conditional distribution of objective variables, we define the following linear combination of the latent variables scaled by the unique variances $\bar{\bm{\eta}}_y$ in order to treat all variables equally:
\begin{align}
    p(\mathbf{y} \mid \mathbf{z}) =& \mathcal{N}(\mathbf{y} \mid \bm{\eta}_y = \bm{\mu}_y + W_{yu} \mathbf{u}, \Psi_y), \\
    \bar{\bm{\eta}}_y \equiv & \Psi_y^{-1/2} \bm{\eta}_y.
\end{align}
Then, the variance of the linear combination of the latent variables $\bar{\bm{\eta}}_y$ can be expressed as follows:
\begin{align}
    \mathrm{Cov}[\bar{\bm{\eta}}_y] =& \mathrm{E}[\Psi_y^{-1/2}(\bm{\eta}_y - \bm{\mu}_y)(\bm{\eta}_y - \bm{\mu}_y)^T \Psi_y^{-1/2}], \notag \\
    =&
    \Psi_y^{-1/2} ( W_{yu} \mathrm{E}[\mathbf{u}\mathbf{u}^T] W_{uy}) \Psi_y^{-1/2}, \notag \\
    =&
    \bar{W}_{yu} \bar{W}_{uy}.
\end{align}
The total variance of $\bar{\bm{\eta}}_y$ is the sum of the diagonal elements of the above covariance matrix,
\begin{align}
    \mathrm{Tr}[\mathrm{Cov}[\bar{\bm{\eta}}_y]] =&
    \mathrm{Tr}[\bar{W}_{yu} \bar{W}_{uy}]
    =
    \mathrm{Tr}[\bar{W}_{uy} \bar{W}_{yu}], \notag \\
    =&
    \mathrm{Tr}[\Omega_{(y,u)}^2]
    =
    \sum_{i=1}^{p_u} [\bm{\omega}_{(y,u)}^2]_i,
\end{align}
and this is the sum of the squared column norms of the matrix $\bar{W}_{yu}$.
Therefore, we define the contribution ratio of the latent space as the proportion of the variation in the denoised objective variables $\bar{\bm{\eta}}_y$ that can be accounted for by the latent space.
The contribution ratio $P_i$ and cumulative contribution ratio $C_i$ are defined as
\begin{align}
    P_i^{(u)} =&
    \frac{[\bm{\omega}_{(y,u)}^2]_i}{\sum_{j=1}^{p_u} [\bm{\omega}_{(y,u)}^2]_j },
    \hspace{0.5cm}
    C_i^{(u)} =
    \frac{\sum_{j=1}^{i}[\bm{\omega}_{(y,u)}^2]_j}{\sum_{j=1}^{p_u} [\bm{\omega}_{(y,u)}^2]_j },
    \hspace{0.5cm} i \in \{1, 2, \dots, p_u \}.
\end{align}
The contribution ratio of the latent space can also be interpreted in terms of the reconstruction error of the scaled regression coefficient $\bar{W}_{yu}$ for $\mathbf{m}^{(x)}$, measured by the squared Frobenius norm~\cite{Arai2022}.

\section{Numerical validation} \label{sec:numerical_validation}
In this section, we first apply the proposed PLS regression to real data.
Then, the estimated parameters are treated as the ground truth, and synthetic datasets are repeatedly generated via random sampling.
By refitting the model to these datasets, we examine the properties of the sampling distribution of the maximum likelihood estimates.

\subsection{Application to HIV protease mutation data}
In this subsection, we numerically validate the proposed PLS regression using a publicly available real dataset, HIV drug resistance data.
We analyzed the mutation of amino acid sequences of Human Immunodeficiency Virus (HIV) type-1.
The dataset was obtained from the HIV drug resistance database~\cite{hiv}.
Details of the database and related datasets can be found in \cite{Rhee2003}.
When an antiretroviral drug is dosed on a patient, the virus becomes resistant to the drug over time by mutating its genes.
This mutation has been observed to be highly cooperative, with each residue in the amino acid sequence mutating not simply stochastically~\cite{Ohtaka2003}.
Although the molecular mechanism of drug resistance has not yet been elucidated, it is expected that the relationship between the correlation pattern of mutations and drug resistance will provide clues to the molecular mechanism of drug resistance.
Among viral resistance to various drugs, we focused on resistance to protease inhibitors.
The data for analysis consists of mutational information on residues in amino acid sequences from position 1 to 99 in protease of viruses isolated from plasma of HIV-1 infected patients, represented by \texttt{P1} to \texttt{P99}, and in vitro susceptibility to various protease inhibitors such as Nelfinavir.
As a preprocessing, the residues in amino acid sequences were converted to binary values, i.e., information on the type of mutation was ignored.
That is, the residues were encoded to $1$ if any mutation, such as insertion, deletion, or substitution from the consensus wild-type amino acid sequence, is present, and encoded to $0$ if there is no mutation from the consensus sequence.
The variables of the drug susceptibility were transformed using a natural logarithmic function to approximately satisfy the normality assumption of the proposed PLS regression.

In this paper, we treated binary variables as if they were continuous, that is, we adopted the method of quantification.
When one adopts the quantification method in factor analysis, two issues have been identified in the previous study~\cite{Arai2022}.
The first issue is inefficient data compression by the latent variables.
Specifically, more numbers of latent dimensions are required to successfully reproduce the correlation structure among observed variables compared to the approach that treats binary variables appropriately.
The second issue arises when the number of latent dimensions is insufficient.
If a dummy variable has low variance, i.e. its mean is below 0.05 or above 0.95, its correlation with other variables tends to be overestimated.
To ensure the validity of the quantification, in this paper, we first performed feature selection by excluding residues in the amid acid sequence with mutation rates below 0.1.
In addition, the number of latent dimensions was selected based on an information criterion.
Therefore, we proceed under the assumption that the quantification is valid.

We employed drug resistances to protease inhibitors as objective variables and the mutations of the amino acid sequence as explanatory variables.
For illustrative purposes, we reduced the number of explanatory and objective variables.
We selected three objective variables, Indinavir ($\texttt{IDV}$), Nelfinavir ($\texttt{NFV}$), and Saquinavir ($\texttt{SQV}$), which had the fewest missing values among all candidate objective variables.
In addition, we selected seven amino acid residues in descending order of correlation with one of the selected objective variables.
These residues are, $\texttt{P54}$, $\texttt{P10}$, $\texttt{P71}$, $\texttt{P46}$, $\texttt{P84}$, $\texttt{P90}$, and $\texttt{P82}$ in order of maximum correlation strength.
To avoid selection bias, we selected a single isolate per patient when multiple isolates were available.
We excluded data points with missing values in the objective variables.
The sample size of the dataset after data cleansing was $N=3264$.
To improve the visibility, the data points in the figures in this section are colored based on the values of the variable $\texttt{IDV}$.

We applied the proposed PLS regression on the HIV protease mutation data.
The model parameters were estimated by maximum likelihood estimation.
Fig.~\ref{fig:likelihood_bic} shows the negative log-likelihood and Bayesian information criterion (BIC)~\cite{bic} as a function of the number of latent dimensions.
From the figure, we observe that the optimal number of latent dimensions, as determined by the minimization of BIC, is $(p_u, p_v) = (3, 5)$.
The likelihood saturates as the number of latent dimensions increases, even in the cells that have been proven to be identifiable.
That is, when the likelihood saturates, the model appears to be overparametrized and thus seems nonidentifiable.
However, when the dimension satisfies $p_u < p_y$, we numerically confirmed that some columns of the factor loading matrix tend to have values close to zero and make only negligible contributions.
As a result, the parameters appear to be identifiable even in cases where identifiability has not been theoretically proven.
Another important observation from the figure is that the same likelihood values appear when $p_u > p_y$.
Specifically, some elements on the anti-diagonal in the likelihood matrix take the same value.
For example, the cells corresponding to $(p_u, p_v) = (3, 5)$ and $(p_u, p_v) = (4, 4)$ in the figure have the same likelihood values.
Since the latter model contains more parameters than the former, this fact suggests that the latter model may be overparametrized and nonidentifiable.
In fact, we numerically confirmed that the models with $p_u > p_y$ are not identifiable.
Such nonidentifiability due to overparametrization can be avoided by selecting the model based on the minimization of the information criterion.
In other words, model selection based on an information criterion selects an identifiable model.

\begin{figure}[htbp]
    \centering \hspace*{-0cm}
    \includegraphics[viewport= 0 0 461.082 162.847, clip]{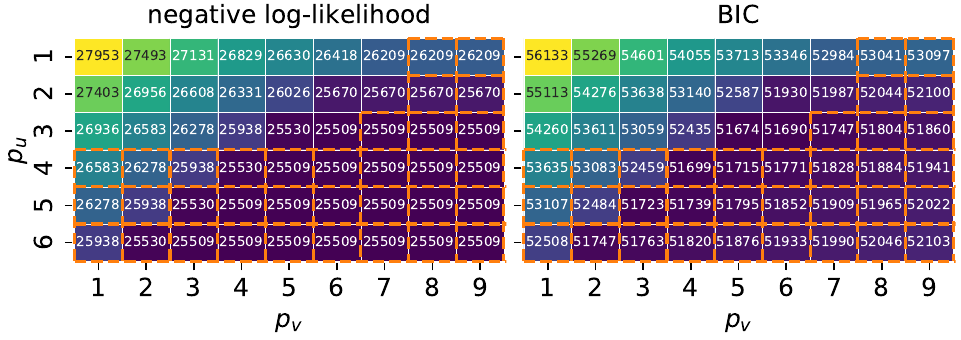}
    \caption{
        Negative log-likelihood and BIC as a function of the number of latent dimensions.
        The cells surrounded by orange dashed lines correspond to models that have not been shown to be identifiable.
    }
    \label{fig:likelihood_bic}
\end{figure}

Next, we compare the predictive performance of the proposed model with that of the previous studies.
As models from previous studies, we adopted the identifiable probabilistic PLS model from the previous study~\cite{Bouhaddani2018} and the classical multi-output PLS model (PLS2) implemented in scikit-learn library for Python ($\texttt{PLSRegression}$)~\cite{sklearn2011}.
To evaluate predictive performance, we first divided the data into training and test sets of equal size.
Since classical PLS regression cannot handle missing values in explanatory variables, we excluded data points with missing values.
Then, we trained the model on the training data, and evaluated the predictive performance using the coefficient of determination on the test data.
The optimal numbers of the latent dimensions selected for the proposed model and the previous identifiable model based on BIC were $(p_u, p_v) = (3, 5)$ and $(p_u, p_v) = (3, 3)$, respectively.
Fig.~\ref{fig:predicted_actual_plot} shows the comparison of predicted and observed values of the objective variables in the proposed model.
In the figure, the coefficient of determination is defined by
\begin{align}
    R^2 = 1 - \frac{\sum_{i=1}^N (y_i - \hat{y}_i)^2}{\sum_{j=1}^N (y_j - \bar{y})^2}, \hspace{1cm} 
    \bar{y} = \frac{1}{N} \sum_{i=1}^N y_i, \hspace{1cm} \hat{y} = \int_{-\infty}^{\infty} y \, p(\mathbf{y} \mid \mathbf{x}) d\mathbf{y},
\end{align}
where $\hat{y}$ is a predicted value of the model.
Table~\ref{table:r_square} presents the coefficient of determination between predicted and observed values for models with different numbers of latent dimensions.
We found that the predictive performance of the proposed model is comparable to that of the classical PLS regression.
In contrast, the identifiable model from the previous study exhibits lower predictive performance compared to these models.
The reason for this is attributed to the constraints imposed on the factor loading matrix, Eq.~(\ref{eq:normalization_previous}).

\begin{figure}[htbp]
    \centering \hspace*{-0cm}
    \includegraphics[viewport= 0 0 451.846 169.799, clip]{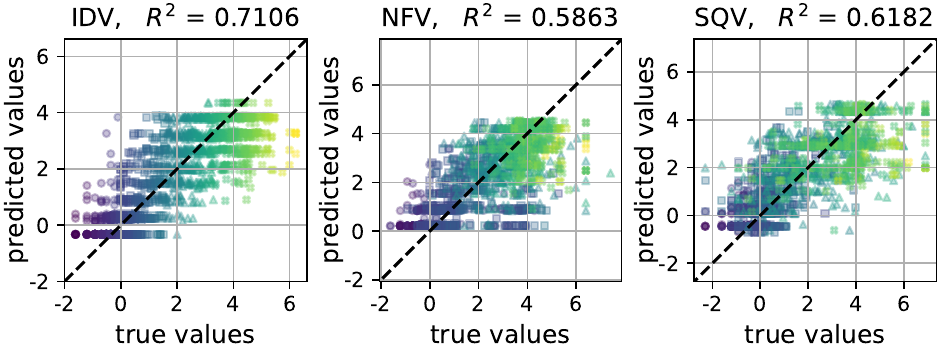}
    \caption{
        Predicted vs. observed plot of the proposed model for three objective variables.
    }
    \label{fig:predicted_actual_plot}
\end{figure}

\begin{table}[hbtp]
    \caption{Coefficient of determination on test data for the proposed model (left), the previous model (center), and the scikit-learn implementation (right), with different numbers of latent dimensions $(p_u, p_v)$ and $\texttt{n\_components}$.}
    \label{table:r_square}
    \begin{minipage}{0.3\linewidth} 
        \centering
        \scriptsize
        \begin{tabular}{c|ccc}
            \hline
            \hspace{0.9cm} & $(1,7)$ & $(2,6)$ & $(3,5)$ \\
            \hline
            $\texttt{IDV}$ & 0.6958 & 0.7079 & 0.7106 \\
            $\texttt{NFV}$ & 0.5819 & 0.5801 & 0.5863 \\
            $\texttt{SQV}$ & 0.5824 & 0.6180 & 0.6182 \\
            \hline
        \end{tabular}
    \end{minipage}
    \hfill
    \begin{minipage}{0.3\linewidth}
        \centering
        \scriptsize
        \begin{tabular}{c|ccc}
            \hline
            \hspace{0.9cm} & $(1,1)$ & $(2,2)$ & $(3,3)$ \\
            \hline
            $\texttt{IDV}$ & 0.7105 & 0.7107 & 0.7066 \\
            $\texttt{NFV}$ & 0.5798 & 0.5750 & 0.5810 \\
            $\texttt{SQV}$ & 0.5338 & 0.5670 & 0.5832 \\
            \hline
        \end{tabular}
    \end{minipage}
    \hfill
    \begin{minipage}{0.3\linewidth}
        \centering
        \scriptsize
        \begin{tabular}{c|ccc}
            \hline
            \hspace{0.9cm} & $(1)$ & $(2)$ & $(3)$ \\
            \hline
            $\texttt{IDV}$ & 0.7117 & 0.7123 & 0.7127 \\
            $\texttt{NFV}$ & 0.5821 & 0.5840 & 0.5868 \\
            $\texttt{SQV}$ & 0.5545 & 0.6140 & 0.6188 \\
            \hline
        \end{tabular}
    \end{minipage}
\end{table}

An important characteristic of probabilistic PLS regression is that it can naturally handle missing values in explanatory variables.
Therefore, we investigated how the predictive performance changes in response to the missing values in training and test datasets.
We randomly introduced missing values into 20 percent of the total entries in the explanatory variable matrix for the training and test datasets described in the previous paragraph, separately.
This procedure was repeated 1000 times and examined how the predictive performance using the coefficient of determination changed.
Fig.~\ref{fig:boxplot_nan} shows the boxplots of the coefficient of determination when missing values (NaN; Not a Number) were introduced, for the proposed and previous identifiable model.
We see that introducing missing values into the training data did not affect the predictive performance while introducing missing values into the test data degraded the predictive performance.
The degree of degradation was similar for both models: the average decrease in the coefficient of determination was 0.0340 for the proposed model and 0.0325 for the previous model when the missing values were introduced into the test data.

\begin{figure}[htbp]
    \centering \hspace*{-0cm}
    \includegraphics[viewport= 0 0 429.083 201.582, clip]{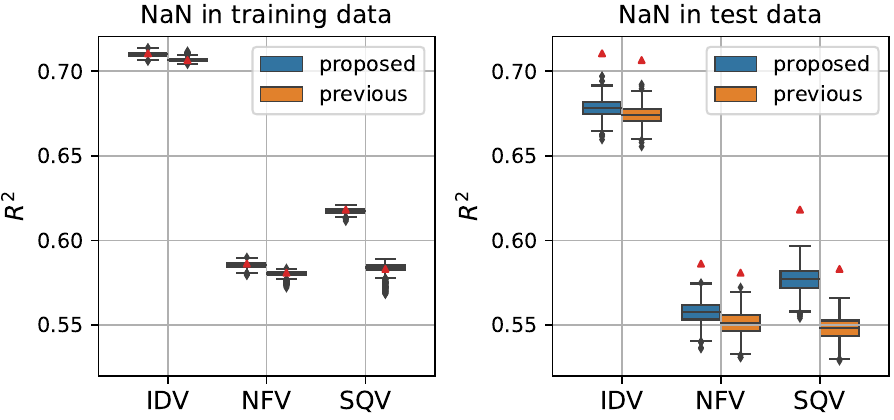}
    \caption{
        Boxplots of coefficient of determination on test data with randomly introduced missing values in training data (left) and in test data (right).
        The red triangles represent the predictive accuracies in the absence of missing values.
    }
    \label{fig:boxplot_nan}
\end{figure}

Another key characteristic of the proposed PLS regression is that it allows for the interpretation of the latent variables, due to the absence of improper solutions and its identifiability.
Fig.~\ref{fig:biplot} shows biplots of the proposed PLS regression for the HIV protease mutation data, where the factor scores are plotted as points and each row of the scaled factor loading matrix is represented as an arrow.
In the biplots, the first and second latent dimensions, which represent the directions with the highest contribution ratios in the shared latent space, are displayed, and the percentages in the axis labels represent the contribution ratio $P_i^{(u)}$.
From the figure, we see that most variations in drug resistance can be explained by the first dimension of the shared latent space.
The variation in objective variables $\texttt{IDV}$ and $\texttt{NFV}$ can be explained primarily by the first latent dimension, whereas the variation in $\texttt{SQV}$ receives additional contributions from the second axis.
This observation explains the fact that increasing the number of shared latent dimensions improves the predictive accuracy for $\texttt{SQV}$.
Although all residues in the amino acid sequence contribute to the first latent dimension, $\texttt{P46}$, $\texttt{P82}$, and $\texttt{P84}$ also contribute to the second latent dimension.
Therefore, these amino acid residues are suggested to contribute to Saquinavir drug resistance.

\begin{figure}[htbp]
    \centering \hspace*{-0cm}
    \includegraphics[viewport= 0 0 468.391 242.411, clip]{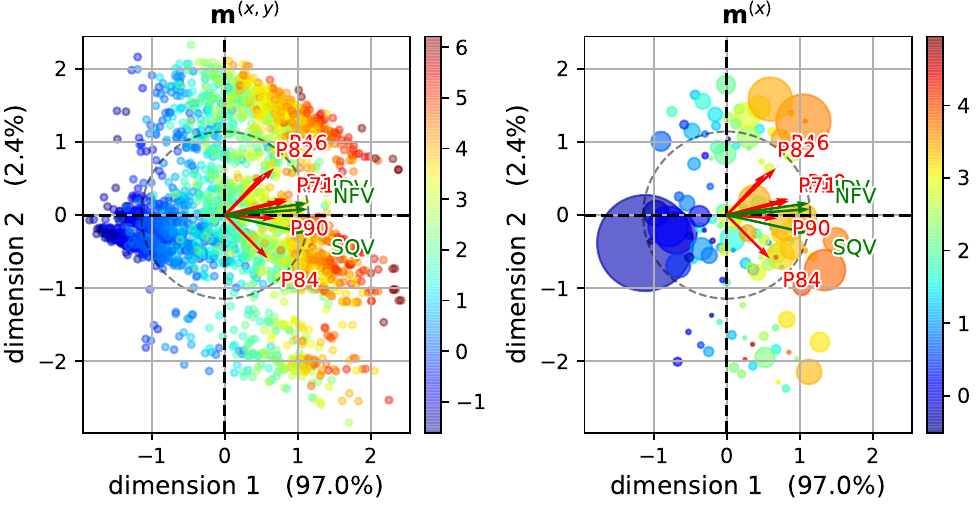}
    \caption{
        Biplots of HIV-1 protease mutation data.
        The markers in the scatter plot are colored from blue to red according to weak to strong drug resistance to the protease inhibitor (Indinavir).
        The areas of the markers are proportional to the sample sizes of corresponding data points.
        The dashed circle represents the maximum possible length of the scaled factor loading vectors.
	    The arrow lengths of the factor loading vectors are scaled to match the points of the factor scores $\mathbf{m}^{(x,y)}$ and $\mathbf{m}^{(x)}$.
    }
    \label{fig:biplot}
\end{figure}

\subsection{Sampling distribution of the maximum likelihood estimates}
In this subsection, we investigate the sampling distribution of the maximum likelihood estimates and examine the consistency of the model parameters.
To investigate the sampling distribution, we generated synthetic data from the model and estimated model parameters using maximum likelihood estimation.
To generate synthetic data, we adopted the model parameters estimated on the HIV protease mutation data in the previous subsection to serve as the ground truth.
We generated 1000 sets of synthetic data with sample sizes of 1000, 3000, and 9000, respectively.
Note that although the residues in the amino acid sequence are inherently discrete variables, the synthetic data generated using the observed joint distribution Eq.~(\ref{eq:likelihood}) are continuous.
We then fitted the model to these datasets using maximum likelihood estimation.
Fig.~\ref{fig:sampling_distribution} shows the sampling distribution of the model parameters for probabilistic PLS regression.
We empirically observed that as the sample size increases, the estimates exhibit consistency and asymptotic normality.

\begin{figure}[htbp]
    \centering \hspace*{-0cm}
    \includegraphics[viewport= 0 0 471.792 285.044, clip]{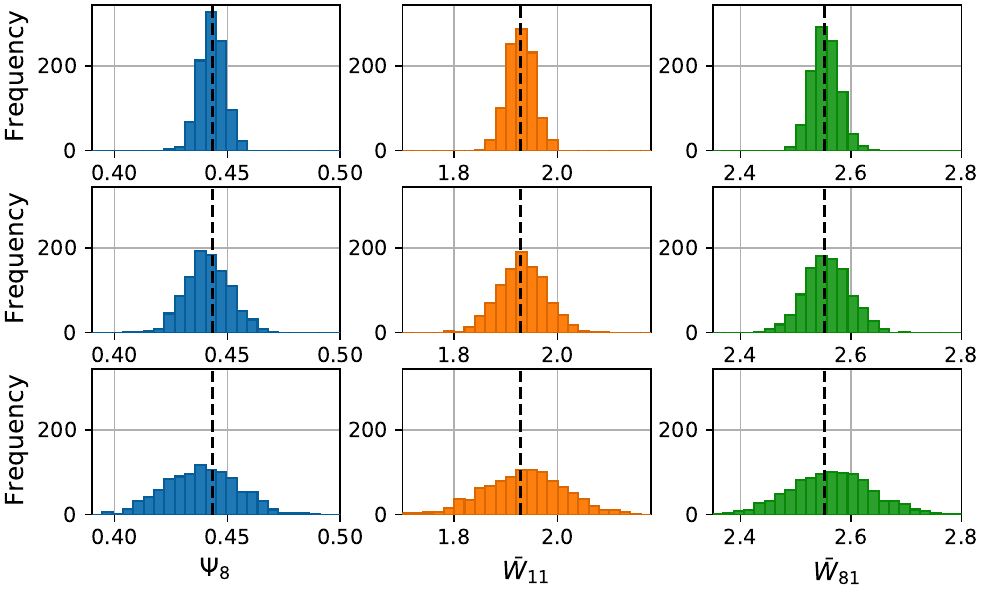}
    \caption{
        Sampling distribution of the maximum likelihood estimates for model parameters.
        The true parameters are indicated by the dashed lines.
        The sample sizes are 1000 for the lower row figures, 3000 for the middle row figures, and 9000 for the upper row figures.
    }
    \label{fig:sampling_distribution}
\end{figure}

\section{Conclusion} \label{sec:conclusion}
We addressed the issues related to improper solutions and identifiability in probabilistic PLS regression.
We found that for the model to be identifiable, it is necessary that the total number of latent dimensions be smaller than the number of observed variables, and that the number of shared latent dimensions does not exceed the number of objective variables.
In particular, when there is only one objective variable, the condition implies that the dimensionality of the shared latent space must be one.
This fact has not been pointed out in the context of classical PLS regression.
We numerically validated the conditions under which the model is identifiable through real data analysis.

Although the issues related to improper solutions and identifiability do not affect the predictive performance of the model itself and the regression coefficients can be estimated without explicitly considering underlying latent variables, extracting insights from these regression coefficients that could lead to an improvement in predictive performance is difficult.
For example, in the presence of multicollinearity in explanatory variables, it is impractical to vary one specific explanatory variable while keeping the other explanatory variables fixed.
In such a situation, interpreting the latent variables in the probabilistic PLS regression and considering to adjust them can provide insights into how to control the objective variables.
Thus, addressing these issues and giving an interpretation of the latent variables is essential for making the model practically useful.

For classical PLS regression, various algorithms have been developed and implemented in various software packages~\cite{sklearn2011, Mevik2007}.
Although the predictive performance of the proposed PLS regression is only comparable to that of the classical PLS regression, the proposed model has identifiable model parameters and the ability to naturally handle missing values in data.
Therefore, the proposed model can be regarded as a superior alternative to the classical PLS regression.
In fact, we analyzed the HIV protease mutation data with missing values and demonstrated that probabilistic PLS regression can effectively handle missing values.
The ability to handle missing values would be particularly useful when applying probabilistic PLS regression to genome-wide association studies, where missing values in single nucleotide polymorphism (SNP) are usually complemented by SNP genotype imputation.
An important direction for future work is to extend probabilistic PLS regression to cases where observed variables include binary and categorical variables.

\appendix

\section{Proof of the identifiability of factor analysis} \label{sec:proof_fa}
In this appendix, we provide a proposition concerning the identifiability of factor analysis and its proof.

\begin{proposition} \label{proposition:proposition_1}
    Let $\mathbf{x}$ be a $p$-dimensional column vector of observed variables.
    We consider factor analysis where the distribution of $\mathbf{x}$ is given as follows:
    \begin{align}
        p(\mathbf{x}) =& \mathcal{N}(\mathbf{x} \mid \bm{\mu}, \Sigma), \\
        \Sigma =& \Psi + W W^T, \\
        \hat{W} =& \mathrm{diag}(\Sigma)^{-1/2} W, \hspace{0.5cm} H^2 = \mathrm{diag}(\hat{W} \hat{W}^T),
    \end{align}
    where $\bm{\mu}$ and $\Sigma$ are parameters of mean and covariance matrix of the normal distribution, respectively, $\Psi$ is a diagonal matrix of unique variances whose elements are nonnegative, $W$ is a $p \times q$ factor loading matrix, and $q$ is a dimension of the latent space.
    If the dimension of the latent space satisfies $0 < q < p$, the row norms of the scaled factor loading matrix $\hat{W}$ have the same value for all features, $\mathrm{diag}(\hat{W} \hat{W}^T) = H^2 = h^2 I$, and their value satisfies $(0 \le h \le 1)$, the symmetric matrix $\hat{W}^T \hat{W}$ is diagonalized and its nonzero nondegenerate eigenvalues are sorted in descending order $\hat{W}^T \hat{W} = \mathrm{diag}(\hat{\bm{\omega}}^2) = \hat{\Omega}^2, \hspace{0.3cm} \hat{\omega}_i > \hat{\omega}_j, \hspace{0.3cm} (i < j)$, and the row sums of the scaled factor loading matrix $\hat{W}$ are nonnegative,
    \begin{align}
        \sum_{i=1}^p \hat{W}_{ij} \ge 0, \hspace{0.5cm} j \in \{1, 2, \dots, q\}, \label{eq:fa_reflection_fix}
    \end{align}
    then $\Sigma_1 = \Sigma_2$ means $\Psi_1 = \Psi_2$ and $W_1 = W_2$.
\end{proposition}

As a first step, we restate the proposition~\ref{proposition:proposition_1} in terms of the scaled factor loading matrix $\hat{W}$.
The covariance matrix $\Sigma$ can be reexpressed by the correlation matrix $\tilde{\Sigma}$ as follows:
\begin{align}
    \Sigma =& \Psi + \mathrm{diag}(\Sigma)^{1/2} \hat{W} \hat{W}^T \mathrm{diag}(\Sigma)^{1/2}, \notag \\
    =& 
    \mathrm{diag}(\Sigma)(I - H^2) + \mathrm{diag}(\Sigma)^{1/2} \hat{W} \hat{W}^T \mathrm{diag}(\Sigma)^{1/2}, \\
    \tilde{\Sigma} \equiv & \mathrm{diag}(\Sigma)^{-1/2} \Sigma \mathrm{diag}(\Sigma)^{-1/2} = I - H^2 + \hat{W} \hat{W}^T. \label{eq:fa_correlation_matrix}
\end{align}
Since the covariance matrix does not depend on the choice of parametrization, the identifiability for $\Sigma$ in the proposition~\ref{proposition:proposition_1} is equivalent to the identifiability for $\tilde{\Sigma}$ as follows; if the correlation matrix,
\begin{align}
    \tilde{\Sigma} = (1 - h^2) I + \hat{W} \hat{W}^T, \label{eq:identifiability_2}
\end{align}
can be expressed in two different ways as $\tilde{\Sigma}_2 = \tilde{\Sigma}_2$, then $h_1 = h_2$ and $\hat{W}_1 = \hat{W}_2$.

We first show, by the following lemma, that the parameter $h$ in Eq.~(\ref{eq:identifiability_2}) does not depend on the choice of parametrization.
\begin{lemma} \label{lemma:lemma_1}
    Let the correlation matrix $\tilde{\Sigma}$ be parametrized as in Eq.~(\ref{eq:identifiability_2}).
    Then for the number of latent dimensions $q < p$, $\tilde{\Sigma}_1 = \tilde{\Sigma}_2$ means $h_1 = h_2$.
\end{lemma}
\begin{proof}[\textbf{\upshape Proof:}]
    Since the matrix $\hat{W} \hat{W}^T$ is symmetric and positive semidefinite, it has a spectral decomposition of the form $\hat{W} \hat{W}^T = U \hat{\Omega}^2 U^T$, where a $p \times p$ matrix $U$ is the eigenvector matrix, constructed by placing the eigenvectors of $\hat{W} \hat{W}^T$ as columns, and satisfies $U^T = U^{-1}$, and $\hat{\Omega}^2$ is a diagonal matrix whose diagonal elements are eigenvalues of $\hat{W} \hat{W}^T$.
    By using the decomposition of the identity matrix $I = U U^T$, the correlation matrix of observed variable $\mathbf{x}$ also has the following spectral decomposition: 
    \begin{align}
        \tilde{\Sigma} = U [ (1 - h^2) I + \hat{\Omega}^2 ] U^T.
    \end{align}
    The above equation implies that the eigenvector matrix of $\hat{W} \hat{W}^T$ is also the eigenvector matrix of $\tilde{\Sigma}$.
    Since the matrix $\hat{W} \hat{W}^T$ is rank-deficient, the eigenvalue diagonal matrix $\hat{\Omega}^2$ contains zero eigenvalues.
    Then, the smallest eigenvalue of the correlation matrix $\tilde{\Sigma}$ is $1 - h^2$.
    The eigenvalues of the correlation matrix do not depend on the choice of parametrization.
    Thus, the parameter $h$ is uniquely determined regardless of the parametrization.
\end{proof}

Therefore, if $\hat{W}_1 \hat{W}_1^T = \hat{W}_2 \hat{W}_2^T$ implies $\hat{W}_1 = \hat{W}_2$, the factor analysis model is identifiable.
The identifiability for $\hat{W}$ is demonstrated by the following lemma.
\begin{lemma} \label{lemma:lemma_2}
    For $p \times q$ matrix $\hat{W}$, if $q \le p$, the symmetric matrix $\hat{W}^T \hat{W}$ is diagonalized and its nonzero nondegenerate eigenvalues are sorted in descending order, and the row sums of the matrix $\hat{W}$ are nonnegative, then $\hat{W}_1 \hat{W}_1^T = \hat{W}_2 \hat{W}_2^T$ means $\hat{W}_1 = \hat{W}_2$.
\end{lemma}
\begin{proof}[\textbf{\upshape Proof:}]
    The diagonalization condition of $\hat{W}^T \hat{W}$ means that each column of $\hat{W}$ is orthogonal.
    We separate the norms of the column vectors of $\hat{W}$ and express it as $\hat{W} = (\hat{W} \hat{\Omega}^{-1}) \hat{\Omega}$, where $\hat{\Omega}$ is a $q \times q$ diagonal matrix with positive diagonal elements, and each column of $\hat{W} \hat{\Omega}^{-1}$ is orthonormalized.
    Then, the lemma~\ref{lemma:lemma_2} is equivalent to the statement: if
    \begin{align}
        (\hat{W}_1 \hat{\Omega}_1^{-1}) \hat{\Omega}_1^2 (\hat{\Omega}_1^{-1} \hat{W}_1^T) = (\hat{W}_2 \hat{\Omega}_2^{-1}) \hat{\Omega}_2^2 (\hat{\Omega}_2^{-1} \hat{W}_2^T),
    \end{align}
    then $\hat{W}_1 = \hat{W}_2$ and $\hat{\Omega}_1 = \hat{\Omega}_2$.
    The above equation is a truncated spectral decomposition of the symmetric positive semidefinite matrix, and the spectral decomposition is unique up to the ordering of eigenvalues and the sign of eigenvectors.
    Therefore, if the nonzero eigenvalues of $\hat{W} \hat{W}^T$ are nondegenerate and sorted in descending order, and the sign of the column of $\hat{W}$ satisfies Eq.~(\ref{eq:fa_reflection_fix}), then the decomposition is unique.
    Therefore, $\hat{W}_1 = \hat{W}_2$.
\end{proof}
This completes the proof of the proposition~\ref{proposition:proposition_1}.

\section{Proof of the identifiability of PLS regression} \label{sec:proof_pls}
In this appendix, we provide a proof of the theorem~\ref{theorem:pls_identifiability}.
As in the case of the proof of the proposition~\ref{proposition:proposition_1}, we express the covariance matrix of the observed variables for the probabilistic PLS regression by the correlation matrix as
\begin{align}
    \tilde{\Sigma} \equiv& \mathrm{diag}(\Sigma)^{-1/2} \Sigma \mathrm{diag}(\Sigma)^{-1/2} =
    \begin{bmatrix} \tilde{\Sigma}_{xx} & \tilde{\Sigma}_{xy} \\ \tilde{\Sigma}_{yx} & \tilde{\Sigma}_{yy} \end{bmatrix}, \notag \\
    =&
    (1 - h^2) I + 
    \begin{bmatrix}
        \hat{W}_{xu} \hat{W}_{ux} + \hat{W}_{xv} \hat{W}_{vx} \; \; \; & \hat{W}_{xu} \hat{W}_{uy} \\
        \hat{W}_{yu} \hat{W}_{ux} & \hat{W}_{yu} \hat{W}_{uy}
    \end{bmatrix}
    =
    (1-h^2) I + \hat{W} \hat{W}^T. \label{eq:pls_correlation_matrix}
\end{align}
Since the covariance matrix does not depend on the choice of the parametrization, identifiability for $\Psi$ and $W$ is equivalent to the identifiability for $h$ and $\hat{W}$.

First, according to the lemma~\ref{lemma:lemma_1}, if $p_u + p_v < p_x + p_y$, the parameter $h$ does not depend on the choice of the parametrization, since it corresponds to the smallest eigenvalue of the matrix $\tilde{\Sigma}$.
Next, we focus on the submatrix $\tilde{\Sigma}_{yy}$.
From the lemma~\ref{lemma:lemma_2}, if $p_u \le p_y$, the matrix $\hat{W}_{uy} \hat{W}_{yu}$ is diagonalized and its nonzero nondegenerate eigenvalues are sorted in descending order, and the sign of the column of $\hat{W}_{yu}$ satisfies
\begin{align}
    \sum_{i=1}^{p_y} [\hat{W}_{yu}]_{ij} \ge 0, \hspace{0.5cm} j \in \{1, 2, \dots, p_u \},
\end{align}
then, $\hat{W}_{yu}$ is uniquely determined.
We now turn to the submatrix $\tilde{\Sigma}_{xy}$.
Multiplying $\hat{W}_{yu} (\hat{W}_{uy} \hat{W}_{yu})^{-1}$ on the right of $\tilde{\Sigma}_{xy}$, it follows that $\hat{W}_{xu}$ is also identifiable.
Finally, we focus on the submatrix $\tilde{\Sigma}_{xx}$.
This reduces to the identifiability of the expression $\hat{W}_{xv} \hat{W}_{vx}$, which can be proved in the same way as the identifiability of the expression $\hat{W}_{yu} \hat{W}_{uy}$.
From the lemma~\ref{lemma:lemma_2}, if $p_v \le p_x$, the matrix $\hat{W}_{vx} \hat{W}_{xv}$ is diagonalized and its nonzero nondegenerate eigenvalues are sorted in descending order, and the sign of the column of $\hat{W}_{xv}$ satisfies
\begin{align}
    \sum_{i=1}^{p_x} [\hat{W}_{xv}]_{ij} > 0, \hspace{0.5cm} j \in \{1, 2, \dots, p_v \},
\end{align}
then, $\hat{W}_{xv}$ is uniquely determined.
This completes the proof of the theorem~\ref{theorem:pls_identifiability}.

\section{Expression with missing values} \label{sec:expression_with_nan}
In this appendix, we provide the expressions of the proposed model in the presence of missing values for the programming implementation of the proposed model.

For this purpose, we first define the index label.
We represent the indices of the matrices $\Sigma$ and $W$ using uppercase labels as follows,
\begin{align}
    \Sigma =
    \begin{bmatrix}
        \Sigma_{XX} & \Sigma_{XY} & \Sigma_{XZ} \\
        \Sigma_{YX} & \Sigma_{YY} & \Sigma_{YZ} \\
        \Sigma_{ZX} & \Sigma_{ZY} & \Sigma_{ZZ} \\
    \end{bmatrix}, \hspace{1cm}
    W = 
    \begin{bmatrix}
        W_{XU} & W_{XV} \\
        W_{YU} & O
    \end{bmatrix}.
\end{align}
Furthermore, we write $X_o$ and $Y_o$ as the index labels of the explanatory and objective variables, respectively, for which observed values are obtained.
Then, we express the index label of the observed variables as $O=(X_o, Y_o)$.
The subvectors corresponding to the parts of the explanatory variables $\mathbf{x}$ and the objective variables $\mathbf{y}$ for which observed values are obtained are expressed using the same index label $O$ as follows,
\begin{align}
    \mathbf{x}_O \equiv \begin{bmatrix} \mathbf{x} \\ \mathbf{y} \\ \mathbf{z} \end{bmatrix}_{X_o}, \hspace{1cm}
    \mathbf{y}_O \equiv \begin{bmatrix} \mathbf{x} \\ \mathbf{y} \\ \mathbf{z} \end{bmatrix}_{Y_o}.
\end{align}

Then, the joint distribution of the observed variables can be expressed as follows,
\begin{align}
    p(\mathbf{x}_{O}, \mathbf{y}_{O}) =& \mathcal{N}([\mathbf{x}_{O}; \mathbf{y}_{O}] \mid [\bm{\mu}_{X_o}; \bm{\mu}_{Y_o}], \Sigma_{OO}).
\end{align}

Next, the predictive distribution given the explanatory variables is as follows:
\begin{align}
    p(\mathbf{y} \mid \mathbf{x}_{O}) =&
    \mathcal{N}(\mathbf{y} \mid \bm{\mu}_{Y | X_o}, \Sigma_{Y | X_o}), \\
    \bm{\mu}_{Y | X_o} =& \bm{\mu}_{Y} + \Sigma_{Y X_o} \Sigma_{X_o X_o}^{-1} (\mathbf{x}_{O} - \bm{\mu}_{X_o}).
\end{align}

Lastly, the conditional distribution of factor scores given the observed variables is as follows:
\begin{align}
    p(\mathbf{z} \mid \mathbf{x}_{O}, \mathbf{y}_{O}) =&
    \mathcal{N}(\mathbf{z} \mid \bm{\mu}_{Z | O}, \Sigma_{Z | O}), \\
    \bm{\mu}_{Z | O} =& \bm{\mu}_Z + \Sigma_{Z O} \Sigma_{OO}^{-1}
    \begin{bmatrix} \mathbf{x}_{O} - \bm{\mu}_{X_o} \\ \mathbf{y}_{O} - \bm{\mu}_{Y_o} \end{bmatrix}.
\end{align}

\bibliographystyle{unsrt}
\bibliography{pls_csda_arxiv}

\end{document}